\documentclass[submission,copyright,creativecommons]{eptcs}
\providecommand{\event}{WWV 2015} % Name of the event you are submitting to
 \usepackage{epsfig}
 \usepackage{pst-grad} % For gradients
 \usepackage{pst-plot} % For axes

\usepackage{caption} % to split caption in multiple lines
\usepackage{graphicx} % for \resizebox etc...
\usepackage{verbatim} % used for hiding the block of text, this package includes comment environment
\usepackage{multicol}
\usepackage{amsthm, amsmath} 
\usepackage{bussproofs} % for proof trees

\usepackage{setspace}
\usepackage{chngpage}  % to adjust table margins
\usepackage[titletoc]{appendix}

\usepackage{thmtools,thm-restate} % to display theorem statement where referred

\usepackage{hyperref} %bibtex url and links
\usepackage{arydshln} % for vertical dashed lines in tables
\theoremstyle{definition}

\newtheorem{lemma}{Lemma}
\newtheorem{conjecture}{Conjecture}
\newcommand{\dff}{\stackrel{\it def}{=}}

\newtheorem{theorem}{Theorem}
\newtheorem{definition}{Definition}
\newcommand{\cond}{{\it if}}

\newcommand{\ltransition}[3]{#1 \, \stackrel{#2}{\longrightarrow}\, #3}
\newcommand{\transition}[3]{#1 \, \stackrel{#2}{\rightarrow}\, #3}
\newcommand{\ttransition}[3]{#1 \, \stackrel{#2}{\Rightarrow}\, #3}

\title{A Calculus of Mobility and Communication for Ubiquitous Computing}

\author{Nosheen Gul
\institute{Department of Computer Science, University of Leicester, England}
\email{ng90@le.ac.uk}
}

\def\titlerunning{A Calculus of Mobility and Communication for Ubiquitous Computing}
\def\authorrunning{N. Gul}

\begin{document}
\maketitle

%*********************ABSTRACT**********************************
\begin{abstract}
We propose a \emph{Calculus of Mobility and Communication (CMC)} for the modelling of mobility, communication and context-awareness in the setting of ubiquitous computing. CMC is an ambient calculus with the \emph{in} and \emph{out} capabilities of Cardelli and Gordon's Mobile Ambients. The calculus has a new form of \emph{global communication} similar to that in Milner's CCS. In CMC an ambient is tagged with a set of ports that agents executing inside the ambient are allowed to communicate on. It also has a new context-awareness feature that allows ambients to query their location. We present reduction semantics and labelled transition system semantics of CMC and prove that the semantics coincide. A new notion of behavioural equivalence is given by defining capability barbed bisimulation and congruence which is proved to coincide with barbed bisimulation congruence. The expressiveness of the calculus is illustrated by two case studies.
\end{abstract}

\section{Introduction}

Mark Weiser envisioned \cite{16mw,5mw} that \emph{ubiquitous computing} provides various computing devices available throughout the physical setting. Ubiquitous computing devices are distributed and could be mobile, and interactions among them are concurrent and often depend on the location of the devices. The idea of context-aware computing has originated in \cite{5mw}. It enables an application to adapt to the changes in its environment and location. Recent advancements in technology have made it possible to detect user's presence or position, or to detect other entities of interest to the user. Therefore, context-awareness and location-awareness have become important features of ubiquitous computing environments.  

In literature, a number of formalisms and languages have been introduced for distributed and concurrent systems. Process algebras are used to model formally concurrent systems. Structural Operational Semantics (SOS) is given as a standard approach of defining the semantics of a system by means of transition rules \cite{9M, SOS1}. Several process calculi were developed to model concurrency, communication and distributed systems, most notably CSP \cite{CSP85}, CCS \cite{9M} and ACP \cite{ACP84}. These process calculi have no primitives to describe certain aspects of behaviours of the ubiquitous computing setting, for example mobility and locations. The idea of mobile code has been formalised by Milner in $\pi$-calculus \cite{10M}. The aforementioned process calculi do not represent directly physical mobility of devices and their locations or surroundings. 

The inspiration for our work comes from several mobile ambient and process calculi that have proved useful in the modelling of mobility, communication and structure of systems. The calculus of \emph{Mobile Ambients, MA} for short, \cite{ambC} is a process calculus for modelling mobile agents over wide-area networks. In MA the ambients represent mobile, nested, computational structures with local communication. Ambients are named terms of the form $n[P]$ where $n$ is a name and $P$ a process.  

In smart indoor settings, spatial organisation is considered an important object for providing communication among various fixed and mobile structures. Despite the advances in the ubiquitous and mobile computing, it is fundamental to formally model physical mobility of devices and interactions among mobile agents that may communicate globally. Communication in such settings could be global, which means that agents may interact with subagents inside other agents. Moreover, the structures in such settings may be mobile, and may need to have knowledge of their current location and surroundings. In order to model such attributes of ubiquitous computing, this paper presents a \emph{Calculus of Mobility and Communication (CMC)}. In CMC, mobility, global communication and location-awareness are considered as first class entities. According to \cite{ambC}, MA was proposed to model mobility and locations that could not be modelled directly by other traditional calculi like Milner's \emph{Calculus of Communicating Systems (CCS)}, therefore we model locations and mobility as in Mobile Ambients. In MA, ambients  may enter or exit named ambients by their $in \; n$ and $out \; n $ capabilities. The ambient's \emph{open} capability dissolves its boundary so that the communication may take place locally. We do not use the \emph{open} capability in this paper since we introduce a new mechanism of global communication. CMC aims at adding global communication as in Milner's CCS. To achieve this we define ambients as $m_A[P]$, where $m$ is the name of the ambient, $A$ is the set of ports that $m$ is allowed to communicate on, and $P$ is an executing agent. This helps in modelling the globally communicating mobile agents in the setting of ubiquitous computing.

We develop Structural Operational Semantics for ambients’ mobility and reuse the CCS rules with an additional rule to introduce global communication. As in MA, we also show ambients’ mobility by means of a reduction semantics. For global communication, in contrast, it would be a challenge to find simple and intuitive reduction rules since (global) communication can happen via an arbitrary number of ambients that could be located far-away in the structure of a term, and it depends on whether or not all these ambients allow the communication. This is however unsound and could derive reductions matching no corresponding transitions. We also develop a new notion of behavioural equivalence for our calculus, and formulate the equivalence in terms of $\alpha$-transitions and observation predicate, inspired by \cite{ambC, BSTSA}. Thus, we define barbed bisimulation and congruence, and capability barbed bisimulation and congruence, we then prove that the respective congruence relations of the two forms of barbs imply each other. The work on behavioural equivalence is still in progress, and the recent advances in the behavioural semantics theory of mobile ambients, as in \cite{BSTSA, ConcretionMA}, should be useful.

Context-awareness is an essential paradigm of ubiquitous computing environment that makes applications adaptive with their surroundings, and enables processes to be aware of the setting in which they are being executed.
%Context-aware applications basically use location of people and computing devices as their main source of contextual information so that the personalised services are executed accordingly.
Therefore, we further extend the syntax of the calculus and add a context-awareness mechanism by introducing two capabilities to it. The new capability \emph{ploc(x).P} allows an ambient to acquire the name of its parent, whereas \emph{sloc(x).P} enquires the sibling's name of an ambient. This feature empowers ambients to have knowledge of their current location and surroundings. 

The rest of the paper is organised as follows: We introduce our basic calculus, CMC$\rm{_b}$, in Section \ref{sec:cmc} where its reduction semantics and labelled transition semantics are given. We show that the two types of semantics coincide for an appropriate sub-calculus of  CMC$\rm{_b}$. In the same section we define behavioural equivalence for the calculus. Section~\ref{sec:iHospital} presents intelligent hospital case study to illustrate the usefulness of CMC$\rm{_b}$. Section \ref{sec:ca} extends CMC$\rm{_b}$ to CMC with the \emph{ploc(x)} and \emph{sloc(x)} capabilities. We present operational semantics for the new capabilities and argue that the semantics coincide. In Section \ref{sec:sm} we give a shopping mall case study to illustrate the usefulness of CMC. Section \ref{sec:conclusion} contains conclusions.

\section{A Calculus of Mobility and Communication} \label{sec:cmc}

We introduce the syntax of the basic part of CMC, denoted by CMC$\rm{_b}$, in Tables \ref{tab:syntx} and \ref{labels}. Informally, CMC$\rm{_b}$ inherits its syntax from MA and CCS. The syntax allows global communication among ambients 
%*****************************SYNTAX OF CMC****************************
\begin{table}[h!]
\centering
$
\begin{array}{llllllllll}
\hline
Names:& m_A,n_B,k_C...\in \mathcal{N}  &  &  \\ 
Actions:& \multicolumn{3}{l}{ \alpha ,\beta,...\in Act  = \mathcal{A} \cup \overline{\mathcal{A}} \cup \lbrace \tau \rbrace }  \\ 
Variables:&  x,y,...\in \mathcal{X} &  &  \\ 
\hline 
Processes:& P,Q ::= \quad D & \vert & C.P & \vert & a(z).P & \vert & \overline{a}(x).P  & 
\\
& \qquad \quad {\vert}  \quad m_{A}[P] & \vert & P+Q  & \vert & P\mid Q  &\vert& (\nu m_A)P 
\\
& \qquad \quad {\vert} \quad (\nu l)P & {\vert} &  P[f]  
\\
Capabilities:& C ::= \qquad x & \vert & \mu & \vert & \epsilon & \vert & C.C' \\
\hline
\end{array}
$
\caption{Syntax of  CMC$\rm{_b}$} 
\label{tab:syntx}
\end{table}
$ $ \\[.5pt]
that could be mobile. We assume that $\mathcal{A}$ is an infinite set of port names, which is ranged over by $a, b, c$, and the set of co-names, denoted by $\overline{\mathcal{A}}$ is ranged over by $\overline{a}, \overline{b}, \overline{c}$. We set $\mathcal{L}=\mathcal{A} \cup \overline{\mathcal{A}}$ and let $A, B, C$ range over it. An infinite set $Act$ comprises all possible actions that an agent can perform and $\alpha, \beta$ range over it. $Act$ also includes $\tau$, which is a single completed action of composite agents. So $Act=\mathcal{L} \cup \lbrace \tau \rbrace$, and the typical subsets of $\alpha$ are $A,B$. The set of agent constants $\mathcal{K}$ is ranged over by $D$ and $E$, and the deadlocked agent 0 is a member of $\mathcal{K}$. In our syntax the variable $z$ in $a(z). P$ can be replaced by a value from a set $\mathcal{V}$, which may contain the capabilities as defined in Table \ref{tab:syntx}. 

For the mobility part of  CMC$\rm{_b}$ syntax in Table \ref{tab:syntx}, we assume an infinite set of ambient names $\mathcal{N}$ that is ranged over by $m_A,n_B$ and $k_C$, where $A,B,C \subseteq \mathcal{A} \cup \mathcal{\overline{A}}$. We define our ambient as a term $m_{A}[P]$, where, $m$ is the name of the ambient, $A$ is the set of ports that ambient $m$ is allowed to communicate on, and $P$ is an executing agent. When ambients allow communication on all visible ports then we shall write $m[P]$ instead of $m_{A}[P]$. Other ambient constructs that are inherited from \emph{MA} are $(\nu m)P$, $C.P$ and $C.C'$. An ambient restriction $(\nu m)P$ executes process $P$ with a private ambient named $m$. In  $C.P$, the process $P$ cannot start execution until the prefix capability $C$ is performed. The capability $\mu$ in Table \ref{labels} allows ambients to perform certain actions, namely $in$ and $out$, whereas $C.C'$ represents a sequence of capabilities (path) when input variable represents one or more of these capabilities. The empty path is represented by $\epsilon$.
 
We further borrow the constructs for agent constants, action prefixing, parallel composition, summation and action restriction from Milner's CCS or the $\pi$-calculus \cite{9M, 10M}. The agent constant $D$ has a unique equation of the form $D\dff P$ where $P$ is an agent that may contain agent constants. The agent constants can also be defined in terms of each other. $\overline{a}(x)$ and $a(z). P$ sends or receives a message on port $\overline{a}$ and $a$ respectively, and then execute $P$. The received message can be any value $v\in V$, and is bound to the variable $z$ in $P$. Parallel composition is given in terms of a binary operator, $P\mid Q$, and summation is given by the choice operator $P + Q$ that allows either process $P$ or process $Q$ to execute. In $(\nu l)P$ the port labels $l$ or $\overline{l}$ are restricted in $P$, where $l \in \mathcal{L}$. In a relabelling $P[f]$, $P$ is a process with the relabelling function $f$ applied to its action labels. Finally, we have the set of terms $T(\Sigma, V)$, where $V$ is the set of process variables, and $T(\Sigma)$, the set of closed terms (agents or processes) ranged over by $P, Q$.
%*****************************************************************
\subsection{Reduction Semantics of CMC$\rm{_b}$}

The reduction semantics is formalised by two concepts: the structural congruence relation, $\equiv$, and the reduction relation $\rightarrow $. We follow the definition in \cite{BSTSA}.
\begin{definition} \label{def:Cont}
A relation $ \mathcal{R}$ over processes in a process calculus is contextual, if it is preserved by all the operators in the process calculus. A relation $ \mathcal{R}$ over processes in a process calculus is p-contextual w.r.t a set of operators \emph{Op}, if it is preserved by all the operators in the set \emph{Op}.
\end{definition}

We denote the set of all names occurring free in $P$ by $\emph{fn}(P)$.

%*****************************************************************
\begin{definition} \label{def:str}
Structural congruence, $\equiv$, over CMC$\rm{_b}$ processes is the least p-contextual equivalence relation w.r.t the set of operators $\emph{Op}_1= \lbrace \nu, \mid, [f],n_B[\;], C. , \alpha. \rbrace$, where $C$ and $\alpha$ are in Tables \ref{tab:syntx} and \ref{labels}, that satisfies the following axioms:
\[
\begin{array}{llllllllll}
P\mid Q \equiv Q\mid P & \rm{(Par Comm)} & \quad A \equiv P \qquad \cond \: A \dff P   & \rm{(Const)}
\\
(P\mid Q)\mid R\equiv P\mid(Q\mid R) & \rm{(Par Assoc)} & \quad (\nu n_B )(P\mid Q) \equiv P \mid (\nu n_B)Q \quad \cond \: n_B \notin \emph{fn}(P) & \rm{(Res Par)}
\\[2pt]
P\mid 0\equiv P & \rm{(Zero Par)} &  \quad (\nu n_B )(m_A[P]) \equiv m_A[(\nu n_B) P] \quad \cond \; n \neq m & \rm{(Res Amb)}
\\
P+Q\equiv Q+P & \rm{(Sum Comm)} &  \quad(\nu n_B) (\nu  m_A) P \equiv (\nu m_A) (\nu  n_B) P & \rm{(Res Res)}
\\
(P+Q)+R \equiv P+(Q+R) & \rm{(Sum Assoc)} & \quad (\nu n_B)0\equiv 0 & \rm{(Zero Res)} \\
P+0 \equiv P & \rm{(Zero Identity )} & \quad  \epsilon.P \equiv P & \rm{(Epsilon)}
\end{array}
\]
\end{definition}

%*****************************************************************
\begin{table}
 $$
\begin{array}{llllllllll}
\hline
\emph{Ambient Prefixes} :& \mu &::=  &in \:n_B  & \vert  & out \: n_B  \\[2pt]

\emph{Action Prefixes}: &  \alpha &::= & a(z)&\vert &\overline{b}(z)& \vert &\tau \\

\emph{Ambient Actions}:& \lambda &::=  &  \emph{enter} \:n_B & \vert & move \:n_B &
\vert  &  exit \: n_B & \vert & \mu \\

\emph{Labels}:& \ell & ::=    & \mu & \vert & \alpha & {\vert}  & \lambda & \vert & \tau  \\

\emph{Outcomes}:& O &::=  & P &  \vert &  K & &
\\
\emph{Concretions}:& K &::=  & (\nu \tilde{m} ) \langle P \rangle Q &   \\
\hline
\end{array}
$$
\caption{Prefixes and labels}
\label{labels}
\end{table}
 %*****************************************************************
\begin{definition} \label{def:red}
The reduction relation, $\rightarrow $, over  CMC$\rm{_b}$ processes is the least p-contextual relation  w.r.t the set of operators $\emph{Op}_2= \lbrace \nu, \mid, n_B[\;] \rbrace $ that satisfies the rule and axioms in Table \ref{tab:redRules}.
\end{definition}
%*****************************************************************
\begin{table}[htbp]
\centering
$
\begin{array}{lllll}
\multicolumn{3}{l}{ m_A[ in \mbox{ } n_B. P \mid Q] \mid n_B[ R] \rightarrow n_B[m_A[P \mid Q]\mid R]} &  & \mbox{(Red In)}
\\
\multicolumn{3}{l}{ n_B[m_A[out \; n_B.P \mid Q]\mid R]  \rightarrow m_A[ P \mid Q] \mid n_B[ R]} & & \mbox{(Red Out)}
\\
\multicolumn{3}{l}{ P \equiv Q, \; Q \rightarrow Q', \; Q' \equiv P' \Rightarrow P \rightarrow P'} & & \mbox{(Red $\equiv$) }
\end{array}
$
\caption{Reduction axioms and rule for  CMC$\rm{_b}$}
\label{tab:redRules}
\end{table}
The axiom Red In in Table \ref{tab:redRules} shows how an ambient $m_A$ may enter into an ambient $n_B$ by the virtue of its $in \; n_B$ capability. The reduction transforms $m_{A}$, which is a sibling ambient of $n_{B}$, into a child of $n_B$. The axiom Red Out describes emigration of an ambient $m_A$ from an ambient $n_B$ by performing the $out \; n_B$ capability. The reduction transforms $m_A$, which is a child of $ n_B$, to a sibling of $n_{B}$.

%*****************************************************************

\subsection{Labelled Transition System for  CMC$\rm{_b}$}
\label{MA:SOS}

A \emph{labelled transition system (LTS)} is a tuple $(S,L, \{\transition{}{l}{}:\; l\in L\})$, where $S$ is a set of states, $L$ is a set of transition labels, and $\transition{}{l}{}$ are transition relations, one for each $l \in L$. The LTS for CMC$\rm{_b}$ is given as follows: The set of processes of CMC$\rm{_b}$ is the set of states, the set of labels $\alpha$ as in Table \ref{labels} is the set of transition labels, and the transition relations $\transition{}{\alpha}$ are defined by Plotkin's \emph{Structural Operational Semantics (SOS)} \cite{SOS1} rules in Tables \ref{rules2} and \ref{rules2:Other}. In our semantics $P \transition {}{\tau} Q$  represents not only binary communication of processes as in CCS but also mobility of ambients by means of their $in \; n_B$ and $out \; n_B$ capabilities. In order to model mobility by $\tau$-transitions additional labels and auxiliary terms are used, namely labels $\lambda$ and concretions $K$ in Table \ref{labels}. So we will need auxiliary transitions $P \transition {}{\lambda} O$, where $P$ is a process, $\lambda$ is a label and $O$ represents an outcome in Table \ref{labels}, which is either a process or \emph{concretion} of the form $(\nu \tilde{m} ) \langle P \rangle Q$ as introduced by Milner \cite{10M} and used by Merro and Hennessy  \cite{BSTSA}. We adopt the following convention after \cite{BSTSA}. If $K$ is the concretion $\nu \tilde{m}  \langle P \rangle Q$, then $\nu u K$ stands for $\nu( u \tilde{m} ) \langle P \rangle Q$, if $u \in fn(P)$, otherwise $\nu \tilde{m}  \langle P \rangle \nu u (Q)$. A similar convention is followed for $\lambda$-Par in Table \ref{rules2:Other}. We define $K \mid P'$ as the concretion $\nu \tilde{m} \langle P \rangle (Q \mid P')$ where using $\alpha$-conversion if necessary, $\tilde{m}$ is selected in such a way that $fn(P') \cap \tilde{m}=\emptyset$. 
%*************SOS RULES FOR CMC (MOBILITY)******************
\begin{table}[h!]
$\begin{array}{lll}
\mbox{(Act) } \displaystyle
 \frac{\;}{\transition{\mu.P}{\mu}{P}} &
 \mbox{(Enter) } \displaystyle
 \frac{\transition{P}{in\:n_B}{P'}}{\transition{{m_{A}[P]}}{\emph{enter}\; n_B}{\langle  m_{A}[P'] \rangle 0}}
\qquad \;
  \mbox{(Co-Enter) } \displaystyle
 \frac{\;}{\transition{n_{B}[P]}{move\; n_B}{\langle P \rangle 0 }}
   \\
\\
\multicolumn{3}{l}{
\mbox{($\tau$-In) } \displaystyle
 \frac{\transition{P}{\emph{enter}\: n_B}{(\nu \tilde{p}) \langle  P' \rangle P''}\quad \transition{Q}{move\: n_B}{(\nu \tilde{q})\langle Q' \rangle Q'' } }{\transition{{P\mid Q}}{\tau}{(\nu \tilde{p})(\nu \tilde{q})(n_{B}[ P' \mid Q'] \mid P'' \mid Q'')}}  \;  ^{(*)}
}\\
\\
\mbox{(Exit) } \displaystyle
 \frac{\transition{P}{out\: n_B}{P'}}{\transition{{m_{A}[P]}}{exit\; n_B}{\langle m_{A}[P'] \rangle 0}}  
  & 
 \qquad \mbox{($\tau$-Out) } \displaystyle
 \frac{\transition{P}{exit\: n_B}{(\nu \tilde{m})\langle P' \rangle P ''}}{\transition{{n_{B}[P]}}{\tau}{(\nu \tilde{m})(P' \mid n_{B}[P''])}} \; ^{(**)}
 	\\
 \end{array}
$
\caption{Transition rules for mobility. Conditions ($*$) and ($**$) are defined as follows:\\[4pt] ($*$) $(fn(P') \cup fn(P'')) \cap  \tilde{q}  =(fn(Q') \cup fn(Q'')) \cap  \tilde{p}  = \emptyset$ and
 ($**$)  $(fn(P') \cup fn(P'')) \cap  \tilde{m}  = \emptyset$}
\label{rules2}
\end{table}
%*****************************************************************
%*************OTHER SOS RULES OF CMC*************
\begin{table}[h!]
%\centering
$\begin{array}{lll}
	 \AxiomC{$P \transition{}{\lambda}{} O  $}
	\LeftLabel{ ($\lambda$-Par)}
	\RightLabel{ $ \; ^{(*)} $}
	\UnaryInfC{$P \mid Q \transition{}{\lambda}{} O \mid Q$}
	\DisplayProof 
	&	
\hspace{1.5cm}	\AxiomC{$P \transition{}{\lambda}{} O  $}
	\LeftLabel{($\lambda$-Res)}
	\RightLabel{ \;($u \notin fn(\lambda)) \; ^{(*)} $}
	\UnaryInfC{$(\nu u)P \transition{}{\lambda}{} (\nu u)O  $}
	\DisplayProof 
\\[5pt]
	\AxiomC{$P \transition{}{\tau}{} P'$}
	\LeftLabel{($\tau$-Amb)}
	\UnaryInfC{$\transition{n_{A}[P]}{\tau}{n_{A}[P']}$}
	\DisplayProof &
\hspace{1.5cm}   \AxiomC{$P\equiv Q \quad \transition{Q}{l}{Q'}\quad Q' \equiv P' $}
	\LeftLabel{(Struct)}
	\UnaryInfC{$\transition{P}{l}{P'}$}
	\DisplayProof
 \end{array}
$ 
\caption{Transition rules for other operators of CMC$\rm{_b}$. Condition ($*$) says that the definition of $\lambda$ is extended to include also a $\tau$.}
\label{rules2:Other}
\end{table}

Transitions $\transition{P} {\lambda} O$ are not first class transitions; they are only helpful in SOS rules that define $\tau$-transitions of processes corresponding to the movement by $in \; n_B$ and $out \; n_B$ capabilities.

Communication in  CMC$\rm{_b}$ is defined as in CCS, so in addition to the SOS rules in Tables \ref{rules2} and \ref{rules2:Other}, we have the SOS rules for CCS as in \cite{9M} (also in Appendix~\ref{Appendix:A}) and the following Global-Com rule:
\begin{prooftree}
\AxiomC{$\transition{P}{\alpha}{P'}$}
	\LeftLabel{(Global-Com)}
	\RightLabel{ $ (\alpha \not = \tau  \mbox { and } \emph{if }(\alpha=a(x)  \emph{ or }  \alpha = \overline{a}(x)  \emph{ for some $a$) then } a \in A)$}
	\UnaryInfC{$\transition{m_{A}[P]}{\alpha}{m_{A}[P']}$}
\end{prooftree}
Global-Com allows ambients to communicate globally only on ports $a \in A$. Recall that when ambients allow communication on all visible channels then we shall write $m[P]$ instead of $m_{A}[P]$. 

Next, we discuss some reductions and at the same time explain how auxiliary labels and transitions are used in defining mobility transitions. We assume $m_{A}[in \: n_B.P] \mid Q \mid n_B[R] $ for some $P$, $Q$ and $R$.

The ambient $m_{A}$, for some $A$, has the capability to enter an ambient $n_B$ for some $B$. By Red In axiom in Table \ref{tab:redRules} we have, 
$$
m_{A}[in \: n_B.P] \mid Q \mid n_B[R] \ltransition{} {}  n_B[m_{A}[P] \mid R] \mid Q.  
$$
We now derive the $\tau$-transition of $m_{A}[in \: n_B.P] \mid Q \mid n_B[R]$ by $\tau$-In rule in Table \ref{rules2}. For simplicity, we assume that there are no private names in $Q$ and $R$. We have $\ltransition {in \; n_B.P} {in \; n_B} P $. When the migration occurs, we must identify the moving ambient $m_A$, and the agent that is left behind. To model these two agents we use concretion $\nu \tilde{m}\langle P\rangle Q$, where $P$ is the agent that moves, while $Q$ is the  agent that stays behind, and $\tilde{m}$ is the set of private names shared by $P$ and $Q$. We introduce a new action $\emph{enter} \; n_B$ and have
$\ltransition{m_A[in \: n_B.P]}{\emph{enter} \; n_B} \langle m_A[P] \rangle 0$. By $\lambda$-Par in Table \ref{rules2:Other} we obtain
$
\ltransition{m_A[in \: n_B.P] \mid Q}{\emph{enter} \; n_B}  \langle m_A[P] \rangle (0 \mid Q) \equiv \langle m_A[P] \rangle Q. 
$ 

Next, to achieve the $\tau$-transition there must exist a sibling ambient $n_B$.  We define a new action $move \; n_B$ for $n_B$ to complete this interaction. By $\tau$-In we get,
\[
\ltransition{m_A[in \: n_B.P] \mid Q \mid n_B[R]}{\tau} n_B[m_A[P] \mid R] \mid Q.
\]
After the transition the ambient $m_{A}$, becomes a child of $n_{B}$. 

Next, we explain emigration capability by considering $m_A[n_B[out \; m_A.P] \mid Q]$, for some $P$ and $Q$ where $Q$ has no private names. The ambient $n_B$ may emigrate from $m_A$ by its $out \; m_A$ capability. By Red Out we have,
\[
 \transition {m_A[n_B[out \; m_A.P] \mid Q]}{} n_B[P] \mid m_A[Q].
\]

We derive the $\tau$-transition of $m_A[n_B[out \; m_A.P] \mid Q]$ by $\tau$-Out. We define a new action $exit \; m_A$, and by Exit in Table \ref{rules2} we get $\ltransition {n_B[out \; m_A.P]} {exit \; m_A} \langle n_B[P] \rangle  0$. By $\lambda$-Par we get $\ltransition {n_B[out \; m_A.P] \mid Q} {exit \; m_A} \langle n_B[P] \rangle  Q$, which shows that when this capability is exercised $n_B[P] $ moves out, while the process $Q$ remains inside $m_A$. By $\tau$-Out we have,
\[
\ltransition { m_A[n_B[out \; m_A.P] \mid Q]} {\tau}  n_B[P]  \mid m_A[Q].
\]
After the transition the ambient $n_{B}$, becomes a sibling of $m_{A}$.

%*****************************************************************
\subsection{Results}
\label{sec:results}
In this subsection we show that the LTS semantics (SOS semantics) coincides with the reduction semantics for a sub-calculus $T'$ of CMC$\rm{_b}$ that consists of all operators of CMC$\rm{_b}$ apart from the prefixing with actions (including $\tau$) operators, the choice operator and the relabelling operator. 

Since we have developed operational semantics for the mobility part of CMC, therefore we intuitively restrict equivalence between the operational semantics to a subset of the calculus, and prove the soundness and completeness of the semantics. Soundness ensures that for every reduction of a $T'$ term there is a valid $\tau$-transition of the term, and the target of the $\tau$-transition is congruent to the target of the reduction. Completeness ensures that for every valid $\tau$-transition of a $T'$ term there is a valid reduction of the term, and the targets of the $\tau$-transitions and the reductions are the same.
%*****************************************************************
\begin{restatable}{theorem}{result}
\label{theo:results}
(a)\;
$\forall P, \;P' \in T'. \; P \rightarrow P' \Longrightarrow \exists \; Q \in T'. \; \transition {P}{\tau} Q \equiv P'$.
% \item[(2)] $\forall P,P' \in T, if \; P \rightarrow P', then  {P}{\tau} Q \; \mbox{ for some } Q \in T$
\quad (b)\; $\forall P,\; R \in T'. \;\transition {P}{\tau} R \Longrightarrow P\rightarrow R$.
\end{restatable}

\begin{proof}
By transition induction where we consider cases of reductions or transitions of terms depending on the structure of the terms. The proof of part \emph{(a)} is straightforward, whereas to show part \emph{(b)} a number of auxiliary lemmas are required, similarly as in \cite{BSTSA}. Given a transition $\transition {P} {\lambda} O $, where $\lambda$ is as in Table \ref{labels}, these lemmas state the structure of terms $P$ and $O$. For example, for ambient entering capability we require the following lemmas:
%if $\lambda = exit \; n_B$ then the lemma for $\lambda$ is:
%\begin{lemma}
%If $ \transition {P}{exit \; n_B } \nu \tilde{m} \; \langle P' \rangle P''$ then  $P \equiv \nu \tilde{m} \; (k_{A}[out \; n_B.P_{1} \mid P_{2}] \mid P_{3}), \; P' \equiv k_{A}[P_{1} \mid P_2] \mbox{ and } P'' \equiv P_3$, for some $  P_{1}, P_{2}, P_{3}, \; k_A \mbox{ with } n_B \not \in  \tilde{m}$, where $\tilde{m}$ is a set of private ambient names in $P$.
%\end{lemma}
\begin{lemma} \label{lemma:enter}
If $ \transition {P}{enter \; n_B } \nu \tilde{p} \; \langle P' \rangle P''$ then  $P \equiv \nu \tilde{p} \; (k_{A}[in \; n_B.P_{1} \mid P_{2}] \mid P_{3}), \; P' \equiv k_{A}[P_{1} \mid P_2] \mbox{ and } P'' \equiv P_3$, for some $  P_{1}, P_{2}, P_{3}, \; k_A \mbox{ with } n_B \not \in  \tilde{p}$, where $\tilde{p}$ is a set of private ambient names in $P$.
\end{lemma}
\begin{lemma} \label{lemma:move}
If $ \transition {Q}{move \; n_B } \nu \tilde{q} \; \langle Q' \rangle Q''$ then  $Q \equiv \nu \tilde{q} \; (n_B[Q_1] \mid Q_2), \; Q' \equiv Q_1 \mbox{ and } Q'' \equiv Q_2$, for some $  Q_{1}, Q_{2},$ with $ n_B \not \in  \tilde{q}$, where $\tilde{q}$ is a set of private ambient names in $Q$.
\end{lemma}

The detailed proof of Theorem \ref{theo:results} is given in \cite{Gul} (Section 4.2).
\end{proof}

%*****************Equivalences*******************************

\subsection{Behavioural Semantics} \label{sec:behSeman}

We develop an appropriate notion of behavioural equivalence for $\rm{CMC_b}$. All processes and context mentioned in this section are from our calculus $\rm{CMC_b}$. We formulate the equivalence in terms of $\alpha$-transitions ($\xrightarrow{\alpha}$), for $\alpha \in  a(z),\overline{b}(z), in\; m_A, out \; m_A, \tau$, for all $a,b,m,A$, and observation predicate as in \cite{ambC, BSTSA}. We write $P\downarrow_{n_A}$ to denote the presence of ambient $n_A$ at the top level, in the other words process $P$ may interact with the environment via $n_A$. We write $P\Downarrow_{n_A}$, if after some number of $\tau$-transitions, the process $P$ exhibits $n_A$ at the top level.

\begin{definition} \textbf{(Barbs)}
\\
%We write $P\downarrow_{n_A}$ to denote the presence of ambient $n_A$ at the top level and process $P$ may interact with the environment via $n_A$. We write $P\Downarrow_{n_A}$, if after some number of $\tau$-transitions, the process $P$ exhibits $n_A$ at the top level. More formally,
%\\[4pt]
$
 P\downarrow{_{n_A}} \dff P \equiv \nu \tilde{m} (n_A[P_1] \mid P_2), \mbox{ where } n_A \not \in \tilde{m} \mbox{ for some} P_1, P_2 
$ and 
\\
$
 P\Downarrow{_{n_A}} \dff  \ttransition{P}{\hat{\tau}} Q \mbox{ and }
Q \downarrow_{n_A} \mbox{ for some } Q.
$
\end{definition}

%A relation $R $ over process $P, Q$ is barb preserving if it is preserved by observation predicates, namely if $P$ may interact with environment via ambient $n_A$ then $Q$ may also interact via the ambient $n_A$ after a number of $\tau$-transitions. Observation predicates of the two process $P,Q$ are invariant under any contexts $\mathcal{C}$[$\cdot$].           
 
\begin{definition} \textbf{(Barb Preserving)}
\\
A relation $R$ over processes is said to be \emph{barb preserving} if $P \; R \; Q$ and $P \downarrow_{n_A} $ implies $Q \Downarrow_{n_A}$.
\end{definition}

\begin{definition} \textbf{(Context)}\\
A \emph{context} $\mathcal{C}$[$\cdot$] is a process with zero or more holes [$\cdot$]. A hole [$\cdot$] in a context $\mathcal{C}$ is replaced by at most one occurrence of a process. A context $\mathcal{C}$[$\cdot$] with a hole [$\cdot$] replaced by a process $P$ is denoted by $\mathcal{C}$[$P$].
\end{definition}

\begin{definition} \textbf{(Contextual Equivalence)} \\
Processes $P$, $Q$ are \emph{contextual equivalent}, denoted by $P \simeq Q$, if for all contexts $\mathcal{C}[\cdot]$ and ambient names $n_A$, $\mathcal{C}[P]\downarrow_{ n_A}$ implies $\mathcal{C}[Q] \Downarrow_{ n_A }$.
\end{definition}

Since we are considering weak equivalence, we provide the notion of weak actions as follows. We write $\alpha \in Act$ (recall that $\tau \in Act$) . 
We write $\Rightarrow$ for the reflexive and transitive closure of $ \overset{\tau}{\rightarrow}  $, where $ \overset{\tau}{\rightarrow}  $ specifies exactly the $\tau$-transition. $ \overset{\tau}{\Rightarrow}  $ specifies at least a $\tau$ transition.  $\hat{\alpha}$ is a sequence obtained by deleting all occurrences of $\tau$ actions, note that $\hat{\tau}=\epsilon$. Furthermore, $ \overset{\hat{\tau}}{\Rightarrow}  $ is $\overset{\epsilon}{\Rightarrow}$, an empty sequence of $\tau$-transitions, and $ \overset{\hat{\alpha}}{\Rightarrow}  $ is $\overset{\alpha}{\Rightarrow}$, for $\alpha \not = \tau$. 

We define two forms of barbs; one at ambient level whereas another for ambients capabilities. They give rise to (a) barbed bisimulation and congruence, and (b) capability barbed bisimulation and congruence. We then show that the respective congruence relations imply each other. %Two processes are barbed congruent if when they are placed into any context then the context processes are barbed bisimilar.

\begin{definition} \textbf{(Barbed Bisimulation and  Congruence)} \label{def:bbc}
\\
A relation $S$ is a \emph{barbed bisimulation}, if it is symmetric and if ($P,Q) \in S$ then for all $\alpha \in \lbrace a(z),\overline{b}(z), \\ in\; m_A, out \; m_A \rbrace$,
\begin{itemize}
\item[-] if $\transition{P} {\alpha} P'$ then $\ttransition{Q} {\widehat{\alpha}} Q'$ and ($P',Q') \in S$;
\item[-] if $P \downarrow_{n_A} $ then $Q \Downarrow_{n_A}$.
\end{itemize}
Processes $P$ and $Q$ are barbed bisimilar, $P \approx Q$, if $(P,Q) \in S$ for some barbed bisimulation $S$. $P$ and $Q$ are barbed congruent, $P \cong Q$, if for all contexts $\mathcal{C}[\cdot]$, $\mathcal{C}[P] \approx \mathcal{C}[Q]$.
\end{definition}

\begin{definition} \label{def:cap}
We write $P \downarrow_{\beta}$ if $\ltransition{P} {\beta} P'$ for some $P'$, where $\beta \in \lbrace in \; n_A, out \; n_A, enter \; n_A, move \; n_A, exit \; n_A \rbrace $. We write $P\Downarrow_\beta$ if $P \xrightarrow{\tau^*}P' \xrightarrow{\beta}P''$ for some $P' $ and $ P''$. 
\end{definition}

%We now define $\beta$-barb bisimulation, where barb congruence between two process remains invariant when they are placed into any context.

\begin{definition} \textbf{(Capability Barbed Bisimulation)}
\\
Let $L= \lbrace in \; n_A, out \; n_A, enter \; n_A, move \; n_A, exit \; n_A \rbrace$, and let $\beta \in L$. A relation $R$ is a \emph{$\beta$-barbed bisimulation}, if $R$ is symmetric and if ($P,Q) \in R$ then for all $\alpha \in \lbrace a(z),\overline{b}(z),in \; n_A, out \; n_A \rbrace$:

\begin{itemize}
\item[-] if $\transition{P} {\alpha} P'$ then $\ttransition{Q} {\widehat{\alpha}} Q'$ and ($P',Q') \in R$;
\item[-] if $P \downarrow_\beta $ then $Q \Downarrow_\beta$.
\end{itemize}
$P$ and $Q$ are $\beta$-barbed bisimilar, $P \approx_\beta Q$, if $(P,Q) \in R$ for some $\beta$-barbed bisimulation $R$. $P$ and $Q$ are barbed congruent, $P \cong_\beta Q$, if for all contexts $\mathcal{C}[\cdot]$, $\mathcal{C}[P] \approx_\beta \mathcal{C}[Q]$.
\end{definition}

%*********************Lemmas needed to prove upcoming Theorem***********
We now prove that two congruence relations, namely barbed bisimulation congruence and capability barbed bisimulation congruence imply each other.

\begin{theorem} \label{theo:equiv}
Let $P,Q$ $\in$ CMC$\rm{_b}$. Then, $ P \cong Q$  iff  $P \cong_\beta Q$ for all $n_B$.
\end{theorem}
\begin{proof}
%The only difference between the two forms of the barbs is the level at which each barb is defined, namely the definition at ambient level and the definition at the capability level. We show that the two forms of barbs imply each other.
We consider a case where $\beta = move \; n_{B}$, and show that $P \cong Q$ implies $ P \cong_{move \; n_B} Q$ for all $P,Q$ and $n_B$.
\\
Assume that $P \cong Q$  and $P \Downarrow_{move \; n_B}$, and we will show $Q \Downarrow_{move \; n_B}$.
We define a context $\mathcal{C}_1$[$\cdot$] as follows:
\\
$ \mathcal{C}_1[\cdot] \dff \nu m_A([\cdot]) \mid \nu a (k_C[in \; n_B.out \; n_B. \overline{a}.0] \mid a.m_A[P]), \mbox{ with }a \not \in B \mbox{ and } a \in C $

%This context allows the ambient $m_A$ to interact with the environment after the communication on $a$ has happened. Also, $k_C$ may execute its $in \; n_B$ and $out \; n_B$ capabilities only if $n_B$ exists in parallel with $k_C$. Therefore, any process replacing the context hole [$\cdot$] must contain $n_B$. Then $k_C$, after executing its capabilities, may communicate on $a$, which enables the context $\mathcal{C}_1$~[$\cdot$] to interact with the environment via $m_A$. 

Global communication is very useful in the definition of context $\mathcal{C}_1[\cdot]$. It acts as a guard and the context may interact with the environment via corresponding guarded ambient if the guard is satisfied. 

Before we continue with proof of Theorem~\ref{theo:equiv}, we shall need the following lemma.

%*************LEMMA FOR LEFT TO RIGHT IMPLICATION OF UPCOMING THEOREM****
%Before we continue with proof of Theorem~\ref{theo:equiv}, we shall need the following lemma.

\begin{lemma} \label{sublemma1}
For $m_A$ and $k_C$ fresh in an agent $R$, $R\Downarrow_{move \; n_B}$ iff $\mathcal{C}_1[R] \Downarrow _{m_A}$.
\end{lemma}

\begin{proof}
We show,
$
R\Downarrow_{move \; n_B} \mbox{ implies } \mathcal{C}_1[R] \Downarrow _{m_A}
$
\\
By Definition \ref{def:cap},
 $R\Downarrow_{move \; n_B}$ \mbox{implies} $\transition{R} {\tau^*}{R'} \xrightarrow {move \; n_B}R''$ for some $R',R''$. Since $R\Downarrow_{move \; n_B}$ is valid, we obtain $R \xrightarrow{\tau*}R'\xrightarrow{move\; n_B} R''$.

We consider ${R'} \xrightarrow {move \; n_B}R''$. By Lemma \ref{lemma:move},
$\mbox{ if } {R'} \xrightarrow {move \; n_B} \nu \tilde{r} \langle Q' \rangle Q''    \mbox{ then }  R'\equiv \nu \tilde{r} (n_B[R_1] \mid R_2) \mbox{ and } R'' \equiv \nu \tilde{r} \langle Q' \rangle Q'',  \mbox{ where } Q' \equiv R_1 \mbox{ and } Q'' \equiv R_2$.
We now have,
$$
\begin{array}{lll}
\mathcal{C}_1[R'] & \equiv  \mathcal{C}_1[\nu \tilde{r} (n_B[R_1] \mid R_2)] \equiv  \nu m_A( \nu \tilde{r} (n_B[R_1] \mid R_2) ) \mid \nu a (k_C[in \; n_B.out \; n_B. \overline{a}.0] \mid a.m_A[P]) \\[4pt]
% & \nu m_A( \nu \tilde{r} (n_B[R_1] \mid R_2) ) \mid \nu a (k_C[in \; n_B.out \; n_B. \overline{a}.0] \mid a.m_A[P]) \\
\end{array}
$$
Since by $(*)$ in $\tau$-In in Table \ref{rules2}, the members of $\tilde{r}$ are not free names in $ \nu a (k_C[in \; n_B.out \; n_B. \overline{a}.0] \mid a.m_A[P])$, and $a \not \in  \emph{fn} (\nu m_A(\nu \tilde{r} (n_B[R_1] \mid R_2)))$, the process $\mathcal{C}_1[\nu \tilde{r} (n_B[R_1] \mid R_2)]$ executes as follows
\[
\begin{array}{lll}
 \xrightarrow{\tau} & \nu a \nu \tilde{r} (\nu m_A(n_B[k_C[out \; n_B. \overline{a}.0] \mid R_1] \mid R_2) \mid  a.m_A[P]),
\\
  & (k_C \not = m_A  \mbox{ and }  k_C \not \in \tilde{r} ) \; \mbox{ and } \; ( a \not \in \emph{fn} (R_2) \mbox{ and }  \tilde{r} \cap \emph{fn} (P)= \emptyset) &  \mbox{($\tau$-In)} 
\\
 \xrightarrow{\tau} & \nu a \nu \tilde{r} ( \nu m_A(n_B[R_1] \mid R_2 \mid k_C[\overline{a}.0] ) \mid   a.m_A[P]) & \mbox{($\tau$-Out)} \\
 \xrightarrow{\tau} & \nu a \nu \tilde{r} ( \nu m_A(n_B[R_1] \mid R_2 \mid k_C[0]) \mid m_A[P]) &  \mbox{(Global-Com)}  \\
\end{array}
\]
We need to show $\mathcal{C}_1[R]\Downarrow_{m_A}$ which by our predicate definition, means $\mathcal{C}_1[R] \xrightarrow{\tau^*} \mathcal{C}_1[R'] \downarrow_{m_A},$ and 
$\mathcal{C}_1[R']\downarrow_{m_A} \mbox{ means }\mathcal{C}_1[R'] \equiv \nu \tilde{m}(m_A[P_1] \mid P_2)$ for some $P_1,P_2, \tilde{m}$.
When $P_2 \equiv \nu m_A(n_B[R_1] \mid R_2 \mid k_C[0]), \; m_A[P_1] \equiv m_A[P]  $ and $ \tilde{m} \equiv \nu a \nu \tilde{r}$, then we obtain $\mathcal{C}_1[R'] \equiv \nu a \nu \tilde{r} (\nu m_A(n_B[R_1] \mid R_2 \mid k_C[0]) \mid m_A[P])$, which implies $\mathcal{C}_1[R']\downarrow_{m_A}$.
\\
Since $ R \xrightarrow{\tau^*} R'$ we obtain $\mathcal{C}_1[R] \xrightarrow {\tau^*} \mathcal{C}_1[R']$. Since  $\mathcal{C}_1[R] \xrightarrow {\tau^*} \mathcal{C}_1[R']$ and $\mathcal{C}_1[R']{\downarrow_{m_A}}$, we obtain $\mathcal{C}_1[R] \Downarrow_{m_A}$ as required.

We now show,
$
 \mathcal{C}_1[R] \Downarrow _{m_A} \mbox{ implies } R\Downarrow_{move \; n_B}
$
\\
Since $\mathcal{C}_1[R] \Downarrow_{m_A} \mbox{ means } \mathcal{C}_1[R] \xrightarrow{\tau^*} \mathcal{C}_1[R']\downarrow_{m_A}$ for some $R'$, we have
$$
 \mathcal{C}_1[R] \equiv \nu m_A(R) \mid \nu a (k_C[in \; n_B.out \; n_B. \overline{a}.0] \mid a.m_A[P]) 
$$
Here, $\mathcal{C}_1[R]$ may interact with the environment via the ambient $m_A$ only if, after some $\tau$-transitions, $m_A$ exists at the top level. Therefore, $R \xrightarrow{\tau *} R'\downarrow_{n_B}$ and we obtain
$$
 \nu m_A(R') \mid \nu a (k_C[in \; n_B.out \; n_B. \overline{a}.0] \mid a.m_A[P])
$$
Since we define predicate ($R' \downarrow_{n_B}$) as
$
R' \downarrow_{n_B} \dff R' \equiv \nu \tilde{q}(n_B[Q_1] \mid Q_2)  \mbox{ for some } Q_1, Q_2, \mbox{and } n_B \not \in \tilde{q},
$
we obtain
$$
 \nu m_A(R') \mid \nu a (k_C[in \; n_B.out \; n_B. \overline{a}.0] \mid a.m_A[P]) \xrightarrow{\tau *} \nu m_A(R') \mid \nu a (k_C[0] \mid m_A[P]).
$$
Since after a number of $\tau$-transitions we have $m_A$ at the top level of context $\mathcal{C}_1$, so $ \mathcal{C}_1[R']$ may interact with environment via $m_A$ and we obtain $\mathcal{C}_1[R'] \downarrow_{m_A}$.
Since $\mathcal{C}_1[R] \xrightarrow{\tau ^*} \mathcal{C}_1[R'] \mbox{ and } \mathcal{C}_1[R'] \downarrow_{m_A}$, we obtain $\mathcal{C}_1[R] \Downarrow_{m_A}$.
\\
Since we have $R' \equiv \nu \tilde{q}(n_B[Q_1] \mid Q_2)$, we show $\xrightarrow{move \; n_B}$ as follows:

\begin{prooftree}
\AxiomC{}
\LeftLabel{Co-Enter}
\UnaryInfC{$\ltransition{n_B[Q_1]} {move \; n_B} \langle n_B[Q_1] \rangle 0 $}
\LeftLabel{$\lambda$-Par}
\UnaryInfC{$\ltransition{n_B[Q_1] \mid Q_2} {move \; n_B} \langle n_B[Q_1] \rangle (0 \mid Q_2) $}
\LeftLabel{$\lambda$-Res}
%\RightLabel{$fn(Q) \not \in \tilde{q}$}
\UnaryInfC{$\ltransition{\nu \tilde{q} ( n_B[Q_1] \mid Q_2)} {move \; n_B}  \nu \tilde{q} \langle n_B[Q_1] \rangle (0 \mid Q_2) \equiv \nu \tilde{q} \langle n_B[Q_1] \rangle  Q_2  $}
\LeftLabel{Struct}
\UnaryInfC{$\ltransition{\nu \tilde{q} ( n_B[Q_1] \mid Q_2)} {move \; n_B}  \nu \tilde{q} \langle n_B[Q_1] \rangle  Q_2 $}
\end{prooftree}
\end{proof}
%*************END OF LEMMA FOR LEFT TO RIGHT PART OF THEOREM**********
 
Now we return to the proof of Theorem~\ref{theo:equiv}.
Since $P \Downarrow_{move \; n_{B}}$ we get, by Lemma~\ref{sublemma1}, $\mathcal{C}_1[P] \Downarrow_{m_A}$. Since $P \cong Q$, we obtain $\mathcal{C}_1[P] \cong \mathcal{C}_1[Q]$, for context $\mathcal{C}_1$[$\cdot$]. Then since $\mathcal{C}_1[P] \cong \mathcal{C}_1[Q] $, $\mathcal{C}_1[P] \Downarrow_{m_A}$ gives us $\mathcal{C}_1[Q] \Downarrow_{m_A}$. Finally, by Lemma~\ref{sublemma1}, $\mathcal{C}_1[Q] \Downarrow_{m_A}$ implies $Q \Downarrow_{move \; n_{B}}$ as required.

Next, we show the right to left implication, namely
$
P \cong_{move \; n_B} Q \Rightarrow P \cong Q \mbox{ for all } P,Q.
$
Assume that $P \cong_{move \; n_B} Q$  and $P \downarrow_{m_A}$, and we will show $Q \Downarrow_{m_A}$.
We define the context $\mathcal{C}_2$[$\cdot$] as follows:\\
$
\mathcal{C}_2[\cdot] \dff \nu n_B([\cdot]) \mid \nu a (k_C[in \; m_A.out \; m_A. \overline{a}.0] \mid a.n_B[P]), \mbox{ with } a \not \in A \mbox{ and } a \in C.
$
%**********BEGIN LEMMA FOR RIGHT TO LEFT IMPLICATION OF THE THEOREM****
\begin{lemma} \label{sublemma2} For $k_C$ and $n_B$ fresh in an agent $R$, $R\Downarrow_{m_A}$  iff  $\mathcal{C}_1[R] \Downarrow _{move \; n_{B}}$.
\end{lemma}

\begin{proof}
Since the proof is very similar to the proof of Lemma~\ref{sublemma1} it is omitted.
\end{proof}
%*************END OF LEMMA FOR RIGHT TO LEFT PART OF THEOREM**********

%Now we return to the proof of right to left implication of  Theorem~\ref{theo:equiv}.
Since $P \Downarrow_{m_A}$ Lemma~\ref{sublemma2} gives us $\mathcal{C}_2[P] \Downarrow_{move \; n_{B}}$. Since $P \cong_{move \; n_{B}} Q$, we obtain $\mathcal{C}_2[P] \cong_{move \; n_{B}} \mathcal{C}_2[Q]$ for context $\mathcal{C}_2$[$\cdot$]. Next, since $\mathcal{C}_2[P] \cong_{move \; n_{B}} \mathcal{C}_2[Q] $, $\mathcal{C}_2[P] \Downarrow_{move \; n_{B}}$ gives us $\mathcal{C}_2[Q] \Downarrow_{move \; n_{B}}$. Hence, by Lemma~\ref{sublemma2}, $\mathcal{C}_2[Q] \Downarrow_{move \; n_{B}}$ implies $Q \Downarrow_{m_A}$ as required.
\end{proof}

\begin{conjecture} We conjecture that Theorem \ref{theo:equiv} will hold for the other capabilities, namely $enter \; n_B$ and $ exit \; n_B$ of CMC$\rm{_b}$.
\end{conjecture}

%The use of global communication in CMC is fundamental in proving the above results.
%*****************************************************************
\section{Intelligent Hospital Case Study} \label{sec:iHospital}

This case study illustrates the usefulness of CMC$\rm{_b}$ in the given problem domain. Agents' mobility and global communication features are modelled in a scenario where services follow mobile ambients, and server supplies services globally to  appropriate device provided that the receiving ambient is at the same location as the device. 

We consider a hospital building and a doctor who moves around the building and helps patients. While dealing with patients, he may need to use information displayed on screens that are fixed around the building. We assume that an independent server communicates globally with the doctor and with the screens around the building. The purpose of this case study is to model services following the doctor around the building, more specifically to ensure that information is shown only on the screens of the rooms where the doctor is present.

An ambient $k$ represents the building. The ambient $k$ contains ambients $\emph{dr}_{K}$ and $w_{L}$ which represent the doctor's room and the ward respectively. $K$ and $L$ are sets of communication ports, where $b, c_{1} \in {K}$ and $b, c_{2} \in {L}$. This means that the ambient $\emph{dr}_{K}$ can communicate at least on ports $a$ and $c_1$ and the ambient $w_{L}$ can communicate at least on ports $a$ and $c_2$. Furthermore, there are two fixed screens $\emph{scr}_{A_{1}}$ and $\emph{scr}_{A_{2}}$ in $\emph{dr}_K$ and $w_L$ respectively. $A_{1}$ and $A_{2}$ are the sets of communication ports, where $c_{1} \in A_{1}$ and $c_{2} \in A_{2}$, but $c_{1} \not \in A_{2}$ and $c_{2} \not \in A_{1}$. Finally, the doctor is represented as an ambient $d_{B}$ for some $B$ with $a,b \in B$.

Initially, the ambient $d_{B}$ is in the doctor's room $\emph{dr}_{K}$, he then moves to the ward $w_{L}$ and starts using services on the screen $\emph{scr}_{A_{2}}$. The graphical representation of our setting is given in Figure \ref{fig:hospital}. The ambients are represented by boxes, whereas dashed lines represent the communication channels. Next, we define our agents as follows:
%***********************INTELLIGENT HOSPITAL FIGURE***************************
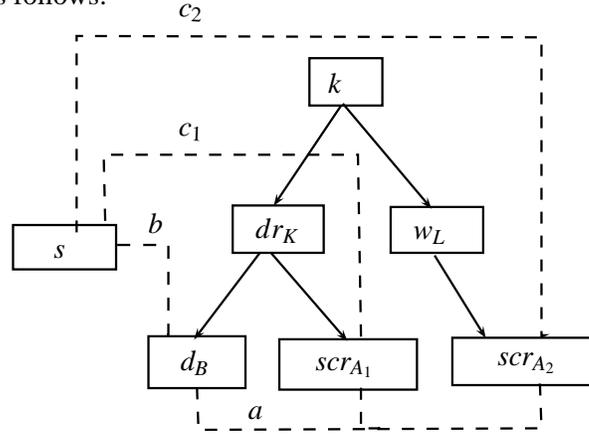
\begin{figure}[htbp]
\centering
\begin{pspicture}(0,-2.8740487)(8.671719,2.35)
\psline[linewidth=0.03cm](8.616718,0.28095123)(8.656718,0.30095124)
\psframe[linewidth=0.03,dimen=outer](4.94,2.1009512)(3.94,1.4409512)
\psframe[linewidth=0.03,dimen=outer](6.28,0.120951235)(5.02,-0.51904875)
\psframe[linewidth=0.03,dimen=outer](4.16,0.14095123)(2.92,-0.51904875)
\psframe[linewidth=0.03,dimen=outer](1.4,-0.11904877)(0.0,-0.7390488)
\psframe[linewidth=0.03,dimen=outer](5.4,-1.6390488)(3.54,-2.3390489)
\psframe[linewidth=0.03,dimen=outer](3.12,-1.5990487)(1.8,-2.319049)
\psframe[linewidth=0.03,dimen=outer](7.82,-1.6190487)(5.84,-2.2790487)
\usefont{T1}{ptm}{m}{n}
\rput(4.2928123,1.7509513){$k$}
\usefont{T1}{ptm}{m}{n}
\rput(0.6328125,-0.48904878){$s$}
\usefont{T1}{ptm}{m}{n}
\rput(3.5328124,-0.20904876){$dr_K$}
\usefont{T1}{ptm}{m}{n}
\rput(5.5528126,-0.26904875){$w_L$}
\usefont{T1}{ptm}{m}{n}
\rput(2.4328125,-1.9690487){$d_B$}
\usefont{T1}{ptm}{m}{n}
\rput(4.4328127,-2.0290487){$scr_{A_{1}}$}
\usefont{T1}{ptm}{m}{n}
\rput(6.852813,-1.9690487){$scr_{A_{2}}$}
\psline[linewidth=0.03cm,arrowsize=0.05291667cm 2.0,arrowlength=1.4,arrowinset=0.4]{->}(4.38,1.4609512)(3.48,0.120951235)
\psline[linewidth=0.03cm,arrowsize=0.05291667cm 2.0,arrowlength=1.4,arrowinset=0.4]{->}(4.4,1.4809512)(5.58,0.080951236)
\psline[linewidth=0.03cm,arrowsize=0.05291667cm 2.0,arrowlength=1.4,arrowinset=0.4]{->}(5.62,-0.5390488)(6.3,-1.6390488)
\psline[linewidth=0.03cm,arrowsize=0.05291667cm 2.0,arrowlength=1.4,arrowinset=0.4]{->}(3.3,-0.49904877)(2.42,-1.6590488)
\psline[linewidth=0.03cm,arrowsize=0.05291667cm 2.0,arrowlength=1.4,arrowinset=0.4]{->}(3.44,-0.49904877)(4.46,-1.6790488)
\psline[linewidth=0.03,linestyle=dashed,dash=0.16cm 0.16cm](1.4,-0.39904878)(2.08,-0.39904878)(2.08,-1.5990487)(2.06,-1.5990487)
\psline[linewidth=0.03,linestyle=dashed,dash=0.16cm 0.16cm](0.86,-0.23166138)(0.86,1.8223927)(0.86,1.8223927)(0.86,2.3809512)(7.04,2.3629332)(7.04,-1.6190487)(7.08,-1.6010307)(7.04,-1.6010307)
\usefont{T1}{ptm}{m}{n}
\rput{-1.8343781}(-0.033701703,0.077399366){\rput(2.3814063,1.1109512){$c_1$}}
\usefont{T1}{ptm}{m}{n}
\rput{-2.699406}(-0.123142436,0.116016835){\rput(2.3814063,2.6909513){$c_2$}}
\usefont{T1}{ptm}{m}{n}
\rput(3.2414062,-2.6290488){$a$}
\usefont{T1}{ptm}{m}{n}
\rput(1.9114063,-0.10904876){$b$}
\psline[linewidth=0.03,linestyle=dashed,dash=0.16cm 0.16cm](1.24,-0.07904877)(1.22,0.80095124)(4.6,0.80095124)(4.64,-1.5790488)(4.64,-1.6190487)
\psline[linewidth=0.03,linestyle=dashed,dash=0.16cm 0.16cm](2.46,-2.2900832)(2.48,-2.8590488)(7.04,-2.838359)(7.02,-2.2590487)(7.0,-2.2590487)
\psline[linewidth=0.03,linestyle=dashed,dash=0.16cm 0.16cm](4.64,-2.2990487)(4.64,-2.8390489)(4.88,-2.8390489)(4.88,-2.8390489)(4.88,-2.819049)
\end{pspicture}
\caption{Intelligent Hospital setting}
\label{fig:hospital}
\end{figure}

%*****************************************************************

Agents \emph{Server} and $S$ are given below, where $l$ is a finite sequence of values, $v_{1}, v_{2},...,v_{k}$, for some $k$:
\begin{center}
$
\begin{array}{llll}
Server (v:l) & \dff & b(x).  \emph{ if } (x=dr_{K} \emph{ then } \overline{c_{1}}(v).Server(l) \qquad %(|l|<n-1)
 \\
& & \emph{ else if } x = w_{L} \; then \; \overline{c_{2}}(v).Server (l) 
   \emph{ else } Server (v:l)) &
\\
 Server(\epsilon) & \dff & 0 \qquad \qquad  S  \; \dff \; s[Server(l)]
 \\
\end{array}
$
\end{center}
%The ambient $s$ outputs the information to the screen around the building on port $c_{m}$ for some $m>0$. We assume that all the screens in the building have unique port to communicate on with the server. Server receives location on port $b$, then sends the information $v$ to the appropriate room.
Agents $\emph{Screen}_m$ and $Scr_{m}$ for $m \in \lbrace 1,2 \rbrace$, are defined as follows:
\begin{center}
$
\begin{array}{llllll}
Screen_{m}& \dff c_{m} (x). \overline{a}(x). Screen _{m} & \quad &
Scr_{m} & \dff &scr_{A_{m}}[Screen_{m}] \\
\end{array}
$
\end{center}
The agent $Scr_{A_{m}}$ receives an input $x$ from the server on $c_{m}$ and outputs $x$ on $a$. Since $a \in B$, the agent $Doc$, defined below, is able to view $x$ via port $a$.

Finally, we define agents \emph{Doctor} and \emph{Doc} as follows:
\begin{center}
$
\begin{array}{lllll}
Doctor(p, l) & \dff& \overline{b}(p). a(x). Doctor(p,x:l) \\
&& + out\: p. \overline{b}(k). in\: r. \overline{b}(r). a(x). Doctor(r,l) && p,r  \in \lbrace dr_K, w_L \rbrace \emph{ and } r \not = p
\\
Doc & \dff& d_{B}[Doctor(dr_K, \epsilon                                            )]
\end{array}
$
\end{center}
We use $p$ to represent the initial location of \emph{Doc}, here $p=dr_K$. When \emph{Doc} leaves $p$ by performing $out \; p$ capability, his new location becomes $k$. He now may enter $r$ by $in \; r$, and send his location to \emph{Server}. In this particular situation, $r=w_L$ since $r \not = p$ and $p=dr_K$.

The Intelligent Hospital system is represented by the parallel composition of the server and the building, which contains doctor's room, ward, the doctor and two screens:  
\begin{center}
$
S  \mid  k[dr_K[Doc  \mid  Scr_1 ]  \mid  w_L[Scr_2] \; ] 
$
\end{center}

For simplicity we assume that the server $S$ sends only a single piece of information, namely $l=v:\epsilon$ for some $v$. Initially \emph{Doc} is in $\emph{dr}_K$ and $S$ wants to send the value $v$ to \emph{Doc} via either $\emph{Scr}_1$ or $\emph{Scr}_2$. There are two possible sequences of execution of the Intelligent Hospital system. These sequences are:

\[
\begin{array}{llll}

 (i) & \ltransition{} {\tau_{b(dr)}} \ltransition{} {\tau_{c_1(v)}}\ltransition{} {\tau_{a(v)}} & \qquad
 (ii) &\ltransition{} {\tau_{out}} \ltransition{} {\tau_{b(k)}} \ltransition{} {\tau_{in}} \ltransition{} {\tau_{b(w)}} \ltransition{} {\tau_{c_2(v)}} \ltransition{} {\tau_{a(v)}} \\
 
\end{array}
 \]
In the first sequence, \emph{Doc} sends its location $\emph{dr}_K$ to $S$ on port $b$, the server in response sends $v$ to $\emph{Scr}_1$ on port $c_1$, and then \emph{Doc} views $v$ via port $a$. These interactions are indicated by appropriate labels that annotate the $\tau s$ of this sequence. In the second case, \emph{Doc} leaves the \emph{dr}$_K$ and enters the ward by its $out\; dr_K$ and $in \; w_L$ capabilities. It sends its current location to $S$ on port $b$ after executing every move capability. The server in response sends $v$ to the $\emph{Scr}_2$ on port $c_2$, and then the screen displays $v$ to $Doc$ on port $a$.

%**********************Context-Awareness-START***********************************

\section{Adding Context-Awareness} \label{sec:ca}

In this section we extend the calculus even further by adding a context-awareness mechanism. In smart indoor settings, location is considered an important entity for providing communication among various portable and static structures. We consider location as one of the most typical forms of context, and propose a location-awareness feature, by introducing new constructs $ploc(x)$ and $sloc(x)$, that query an ambient's parent and sibling names respectively.

We add \emph{ploc(x)} and \emph{sloc(x)} to CMC$\rm{_b}$, finally giving our full calculus CMC. The definition of $\mu$ in Table \ref{labels} is extended to include further \emph{ploc(x)} and \emph{sloc(x)}. Also, the definition of $\lambda$ in Table \ref{labels} is extended to include further auxiliary labels \emph{ploc1(z)}, \emph{sloc1(z)} and \emph{amb $n_B$}.

\subsection{Reduction Semantics for CMC}

The reduction semantics of CMC is given in terms of the structural congruence relation, $\equiv$, and the reduction relation, $\rightarrow $. The axioms for $ploc(x)$ and $sloc(x)$ are given in Table \ref{tab:RedCA}. 

\begin{table}[htbp]
\centering
$
\begin{array}{lllll}
\multicolumn{3}{l}{ m_A[n_B[ploc(x).P \mid Q] \mid R] \rightarrow m_A[n_B[P \lbrace x \leftarrow m_A \rbrace \mid Q] \mid R]} & & \mbox{(Red Ploc) } 
\\
\multicolumn{3}{l}{ m_A[P] \mid n_B[sloc(x).Q \mid S] \rightarrow m_A[P] \mid n_B[Q  \lbrace x \leftarrow m_A \rbrace \mid S]} & & \mbox{(Red Sloc) }
\end{array}
$
\caption{Reduction axioms for \emph{ploc} and \emph{sloc}}
\label{tab:RedCA}
\end{table}

Structural congruence, $\equiv$, for CMC processes is as in Section \ref{def:str} where capabilities $C$ include additionally \emph{ploc(x)} and \emph{sloc(x)}. The reduction relation, $\rightarrow $, for CMC processes is as in Definition~\ref{def:red} except that it satisfies additionally the axioms in Table \ref{tab:RedCA}.

\subsection{SOS Semantics for ploc and sloc}

The SOS rules for $ploc(x)$ and $sloc(x)$ in Tables \ref{tab:PlocSOS}, \ref{tab:SlocSOS} and \ref{rules2:Other}. As before, we use concretions in our rules. We illustrate reductions and transitions associated with the \emph{ploc(x)} capability by considering $m_A[n_B[ploc(x).P_1 \mid P_2] \mid Q]$, where $n_B$ is a child of $m_A$ and $P_1$, $P_2$ and $Q$ are some processes. We assume for simplicity that $P_2$ and $Q$ have no private names. The construct $ploc(x)$ enables $n_B$ to find out the name of its parent (here $m_A$) and pass it to $P$ via $x$. By the reduction rule Red Ploc we get
$
m_A[n_B[ploc(x).P_1 \mid P_2] \mid Q]  \longrightarrow   m_A[n_B[P_1 \lbrace x \leftarrow m_A\rbrace \mid P_2] \mid Q],
$
where $P_1 \lbrace x \leftarrow m_A \rbrace $ denotes process $P_1$ with all occurrences of $x$ replaced by $m_A$.

Now we show how to use $\tau$-Ploc in Table \ref{tab:PlocSOS}. The $\tau$-Ploc rule uses the notion of \emph{lookahead} as, for example, in \cite{Ulidowski92}. In order to derive a $\tau$-transition of $m_A[P]$ we need to ensure that $P$ contains an ambient enquiring parent's name. This is achieved by  $\ltransition{P} {ploc1(z)}  \nu \tilde{p} \langle P' \rangle Q$ where $P'$ contains this ambient. The agent $P'$ then perform $ploc(z)$ to substitute the parent's name: $\ltransition{P'} {ploc(z)} P''$. Hence $P'$ is used both on the right-hand side and on the left-hand side of the premises in $\tau$-Ploc, so $\tau$-ploc has a lookahead.

 Now, to derive the $\tau$-transition of $m_A[n_B[ploc(x).P_1 \mid P_2] \mid Q] $ we must identify the ambient enquiring parent's name. To achieve this we introduce a new action \emph{ploc1(z)} and by Ploc1 we obtain $ \ltransition {n_B[ploc(x).P_1 \mid P_2]} {ploc1(z)} \langle n_B[ploc(x).P_1 \mid P_2] \rangle 0$. By Par-Ploc1 we have
\begin{equation} \label{PlocPrem1}
\tag {A} \ltransition {n_B[ploc(x).P_1 \mid P_2] \mid Q} {ploc1(z)} \langle n_B[ploc(x).P_1 \mid P_2] \rangle ( 0 \mid Q) \mbox{, where }  z \not \in fn(Q) 
\end{equation}
Transition \ref{PlocPrem1} matches the first premise of $\tau$-Ploc for the agent $m_A[n_B[ploc(x).P_1 \mid P_2] \mid Q]$. \\
Now, $n_B[ploc(x).P_1 \mid P_2]$ must be able to perform the capability \emph{ploc(x)}, thus giving the right-hand side premise of $\tau$-Ploc:

\begin{equation} \label{PlocPrem2}
\tag{B} \ltransition {n_B[ploc(x).P_1 \mid P_2] } {ploc(z)}  n_B[P_1 \lbrace x \leftarrow z \rbrace \mid P_2] \mbox{, where }  z \not \in fn(P_2)
\end{equation}
Since we have \ref{PlocPrem1} and \ref{PlocPrem2}, by $\tau$-Ploc we obtain

$
\begin{array}{llll}
 m_A[n_B[ploc(x).P_1 \mid P_2] \mid Q] & \transition{} {\tau} &  m_A[n_B[P_1 \lbrace x \leftarrow z \rbrace \mid P_2] \lbrace z \leftarrow m_A \rbrace \mid 0 \mid Q]\\[4pt]
  & \equiv &  m_A[n_B[P_1 \lbrace x \leftarrow z \rbrace \mid P_2] \lbrace z \leftarrow m_A \rbrace \mid Q]
\end{array}
$
\\
Since $z$ does not appear free in $P_2$ by rules for substitution, the transition is as required:
$
\ltransition {  m_A[n_B[ploc(x).P_1 \mid P_2] \mid Q]} {\tau} m_A[n_B[P_1 \lbrace x \leftarrow m_A \rbrace \mid P_2]  \mid Q]
$
%**********PLOC***********
\begin{table}[t!]
\begin{tabular}{llll} 
\multicolumn{4}{l}{      
        \AxiomC{}
        \LeftLabel{(Act-Ploc)}
         \RightLabel{ \;($z$ doesn't appear in $P$)}
        \UnaryInfC{$\ltransition{ploc(x) . P }{ploc(z)}P \lbrace x\leftarrow z \rbrace $} \DisplayProof } \\[4pt]
        
	 \AxiomC{}
	\LeftLabel{(Ploc1)}
	\UnaryInfC{$ \ltransition{n_B[P]}{ploc1(z)} \langle n_B[P] \rangle 0$}
	\DisplayProof 
	&   
		\AxiomC{$ \ltransition{P}{ploc1(z)}{ \nu \tilde{p} \langle P' \rangle Q \qquad \ltransition{P'} {ploc(z) } P''}$ }
	\LeftLabel{ ($\tau$-Ploc)}
	\RightLabel{$^{(**)}$}
    \UnaryInfC{$\transition{m_A[P]} {\tau} ( \nu \tilde{p} ) m_{A}[P'' \lbrace z \leftarrow m_A \rbrace \mid Q ])$ }
        \DisplayProof       	
 \end{tabular}
\caption{SOS rules for \emph{ploc}. Condition ($**$) is as in Table \ref{rules2}} 
\label{tab:PlocSOS}
\end{table}
%*****************************************************************

We now consider the correspondence of the reduction semantics and the operational semantics for CMC. Let $T'''$ be  a sub-calculus of CMC that consists of all operators of CMC apart from the prefixing with actions (including $\tau$) operators, the choice operator and the relabelling operator. We easily have the soundness part of this correspondence between the two semantics:

\begin{theorem} \label{theo:CMCpp}
$\forall P, \;P' \in T'''. \; P \rightarrow P' \Longrightarrow \exists \; Q \in T'''. \; \transition {P}{\tau} Q \equiv P'$.
\end{theorem}

We conjecture that the completeness part of the correspondence between the operational semantics and reduction semantics is also valid. The proof relies on several auxiliary lemmas. For example, if $\lambda=ploc1(z)$ then the lemma for $\lambda$ is:

\begin{lemma}

If $\ltransition {P}{ploc1(z)} (\nu \tilde{p}) \; \langle P' \rangle P''' $ and $ \ltransition {P'} {ploc(z)} {P''}$, where variable $z$ does not appear in $P$, then $P \equiv \nu \tilde{p} (n_B[ploc(x).P_{1} \mid P_{2}] \mid P_3), \; P'  \equiv n_B[ploc(x).P_1 \mid P_{2}], \; P''' \equiv P_3 \mbox{ and } P'' \equiv n_B[P_1 \lbrace x \leftarrow z \rbrace \mid P_{2}] $ for some $P_1$, $P_2$, $P_3$, $n_B$ with $n_B \not \in \tilde{p}$, $z \not \in fn(P_2) \cup fn(P_3)$ and $\tilde{p}$ a set of private ambient names in $P$.

\end{lemma}

%***************************SLOC SOS ****************************
\begin{table}[t!]
\begin{tabular}{llll} \\
\multicolumn{4}{l}
{      
        \AxiomC{}
        \LeftLabel{(Act-Sloc)\;}
         \RightLabel{ \;  ($z$ doesn't appear in $P$) }
        \UnaryInfC{ $\ltransition{sloc(x). P }{sloc(z)}P \lbrace x\leftarrow z \rbrace $} \DisplayProof } 
        \\[4pt]
	 \AxiomC{}
	\LeftLabel{(Sloc1)}
	\UnaryInfC{$ \ltransition{m_A[P]}{sloc1(z)} \langle m_A[P] \rangle 0$}
	\DisplayProof 
     &
	\AxiomC{}
        \LeftLabel{(Sib-Amb)}
        %\RightLabel{($z \not \in \emph{fn}(Q)$)}
        \UnaryInfC{$\ltransition{n_B[P] }{amb \; n_B}  P $}
        \DisplayProof 
     \quad 
     	\AxiomC{$\ltransition{P}{amb \; n_B} P'$}
        \LeftLabel{(Par-Amb)}
        %\RightLabel{($z \not \in \emph{fn}(Q)$)}
        \UnaryInfC{$\ltransition{P \mid Q }{amb \; n_B} P' $}
        \DisplayProof 
        \\[4pt]     
	\multicolumn{2}{l}{
		\AxiomC{$ \ltransition{P}{sloc1(z)}{ \nu \tilde{p} \langle P' \rangle P'''}  \quad \ltransition{P'} {sloc(z) } P'' \quad \transition {Q} {amb \; n_B} Q'  $ }
	\LeftLabel{ ($\tau$-Sloc)}
	\RightLabel{$^{(*)}$}
    \UnaryInfC{$\transition{P \mid Q} {\tau}  \nu \tilde{p} ( P'' \lbrace z \leftarrow n_B \rbrace \mid P''') \mid Q$ }
        \DisplayProof 
        }          	
 \end{tabular}
\caption{SOS rules for \emph{sloc}. Condition ($*$) is as in Table \ref{rules2}} 
\label{tab:SlocSOS}
\end{table}

\section{Interactive Shopping Mall Case Study} \label{sec:sm}

This case study illustrates the usefulness of global communication, $in \; n_B$, $out \; n_B$, and \emph{ploc(x)} features of CMC. The shopping mall consists of a number of retail outlets, clients and personal digital assistants (PDAs). To offer clients a high level of services, there is a server that delivers services to clients on requests via PDAs which are distributed inside the mall. The tree representation of the shopping mall setting is given in Figure \ref{fig:mall}, where the initial setting is given on the left-hand side and the final setting is on the right hand side.
\begin{figure}[htbp]
\centering
\begin{tabular}{l:l}
\begin{pspicture}(0,-2.605)(5.64,2.5)
\psframe[linewidth=0.03,dimen=outer](4.58,2.80375)(3.44,2.12375)
\psframe[linewidth=0.03,dimen=outer](5.64,0.90375)(4.26,0.22375)
\psframe[linewidth=0.03,dimen=outer](3.9,0.88375)(2.58,0.20375)
\psframe[linewidth=0.03,dimen=outer](1.6,0.84375)(0.0,0.26375)
\psframe[linewidth=0.03,dimen=outer](4.48,-0.73625)(3.18,-1.37625)
\psframe[linewidth=0.03,dimen=outer](2.64,-0.65625)(1.32,-1.33625)
\usefont{T1}{ptm}{m}{n}
\rput(3.9212499,2.45375){$sm$}
\usefont{T1}{ptm}{m}{n}
\rput(0.7912501,0.55375){$server$}
\usefont{T1}{ptm}{m}{n}
\rput(2.9912498,0.51375){$m$}
\usefont{T1}{ptm}{m}{n}
\rput(4.94125,0.57375){$n$}
\usefont{T1}{ptm}{m}{n}
\rput(1.8712503,-1.06625){$client$}
\usefont{T1}{ptm}{m}{n}
\rput(3.8012505,-1.04625){$pda$}
\psline[linewidth=0.03cm,arrowsize=0.05291667cm 2.0,arrowlength=1.4,arrowinset=0.4]{->}(3.88,2.16375)(3.26,0.84375)
\psline[linewidth=0.03cm,arrowsize=0.05291667cm 2.0,arrowlength=1.4,arrowinset=0.4]{->}(4.08,2.14375)(5.02,0.88375)
\psline[linewidth=0.03cm,arrowsize=0.05291667cm 2.0,arrowlength=1.4,arrowinset=0.4]{->}(2.92,0.20375)(2.28,-0.67625)
\psline[linewidth=0.03cm,arrowsize=0.05291667cm 2.0,arrowlength=1.4,arrowinset=0.4]{->}(2.58,-1.11625)(3.2,-1.13625)
\psline[linewidth=0.03,linestyle=dashed,dash=0.16cm 0.16cm](2.4383333,-1.31625)(2.42,-2.03625)(3.96,-2.03625)(3.98,-1.29625)
\psline[linewidth=0.03,linestyle=dashed,dash=0.16cm 0.16cm](0.86,0.86375)(0.84,1.60375)(4.08,1.60375)(4.06,-0.63625)
\usefont{T1}{ptm}{m}{n}
\rput(3.1042187,-1.76625){$a$}
\usefont{T1}{ptm}{m}{n}
\rput(2.3942187,1.81375){$b$}
\usefont{T1}{ptm}{m}{n}
\rput(2.3442185,-2.66625){$c$}
\psline[linewidth=0.03,linestyle=dashed,dash=0.16cm 0.16cm](0.82,0.32375)(0.84130174,-2.37625)(4.4,-2.31625)(4.399191,-1.31625)
\end{pspicture}    \qquad & \qquad
\begin{pspicture}(0,-2.19)(5.88,2)
\psframe[linewidth=0.03,dimen=outer](3.92,1.96)(2.78,1.28)
\psframe[linewidth=0.03,dimen=outer](4.98,0.06)(3.6,-0.62)
\psframe[linewidth=0.03,dimen=outer](3.24,0.04)(1.92,-0.64)
\psframe[linewidth=0.03,dimen=outer](1.48,0.04)(0.0,-0.62)
\psframe[linewidth=0.03,dimen=outer](3.28,-1.3)(1.98,-1.94)
\psframe[linewidth=0.03,dimen=outer](5.0,-1.28)(3.68,-1.96)
\usefont{T1}{ptm}{m}{n}
\rput(3.335625,1.61){$sm$}
\usefont{T1}{ptm}{m}{n}
\rput(0.72562504,-0.29){$server$}
\usefont{T1}{ptm}{m}{n}
\rput(2.405625,-0.33){$m$}
\usefont{T1}{ptm}{m}{n}
\rput(4.3556247,-0.27){$n$}
\usefont{T1}{ptm}{m}{n}
\rput(4.3256254,-1.63){$client$}
\usefont{T1}{ptm}{m}{n}
\rput(2.6156256,-1.59){$pda$}
\psline[linewidth=0.03cm,arrowsize=0.05291667cm 2.0,arrowlength=1.4,arrowinset=0.4]{->}(3.22,1.32)(2.6,0.0)
\psline[linewidth=0.03cm,arrowsize=0.05291667cm 2.0,arrowlength=1.4,arrowinset=0.4]{->}(3.42,1.3)(4.36,0.04)
\psline[linewidth=0.03cm,arrowsize=0.05291667cm 2.0,arrowlength=1.4,arrowinset=0.4]{->}(4.26,-0.62)(4.28,-1.3)
\psline[linewidth=0.03cm,arrowsize=0.05291667cm 2.0,arrowlength=1.4,arrowinset=0.4]{->}(3.72,-1.56)(3.26,-1.58)
\end{pspicture}   \\
\end{tabular}
\caption{Interactive Shopping Mall settings}
\label{fig:mall}
\end{figure}
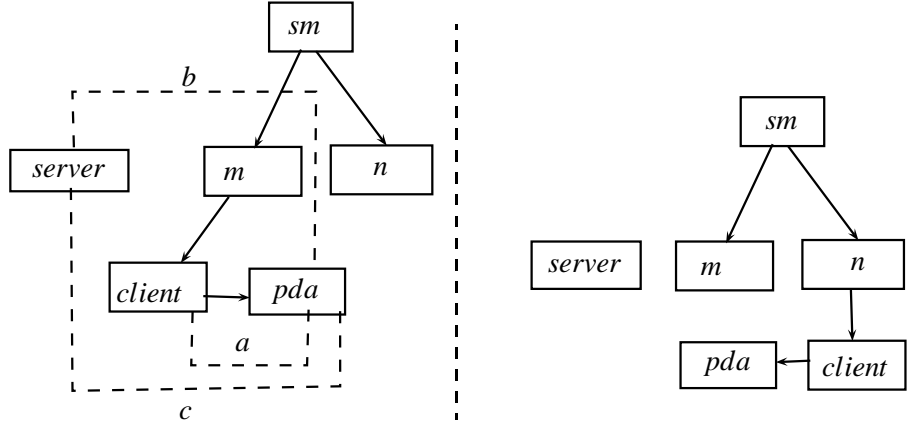
In this figure, the ambient \emph{sm} is the shopping mall with two retail outlets $m$ and $n$. For simplicity we have only one client and one PDA, represented by the ambients \emph{client} and \emph{pda} respectively, which are inside $m$.
\\[2pt]
\textbf{Scenario}: The client wishes to move from her current location $m$ to a target location $n$ inside the mall. She picks up a \emph{pda} and sends the two locations to the $server$ and requests for the path from $m$ to $n$. The server generates this path as a sequence of capabilities and delivers it to the \emph{client} via \emph{pda}.

We define our setting as follows where $C'$, $P'$ and $S'$ are some processes:
%$$
%\begin{array}{llll}
%\nu abc \;(sm[m[\emph{client}[\emph{pull(client)} \; \emph{pda.ploc(x)}.\overline{a}(x,n).a(u).u.C'] \mid 
%\emph{pda}[a(y_1,y_2).\overline{b}(y_1,y_2).c(z).\overline{a}(z).P']] \mid n[\;]]  \\[4pt]
% \quad \mid \emph{server}[b(x_1, x_2). \overline{c}(path(x_1,x_2)).S'] )
%\end{array}
%$$
%The ambient \emph{client} initiates an interaction with the PDA by its \emph{pull(client) pda} capability. After the resulting $\tau$-transition, \emph{pda}, the sibling of \emph{client}, becomes a child of \emph{client}, namely
\begin{center}
$
\begin{array}{llll}
\nu abc \;(sm[m[\emph{client}[\emph{ploc(x)}.\overline{a}(x,n).a(u).u.C' \mid 
\emph{pda}[a(y_1,y_2).\overline{b}(y_1,y_2).c(z).\overline{a}(z).P']] ] \mid n[\;]] \mid \\[4pt]
 \quad server[b(x_1, x_2). \overline{c}(\emph{path}(T,x_1,x_2)).S'] )
\end{array}
$
\end{center}
Here, \emph{path}(T,x$_1,x_2$) is a function that calculates a path between the source node $x_1$ and the target node $x_2$ in a given tree $T$. The only possible execution sequence from this state is $\ltransition{} {\tau_{ploc}} \ltransition{} {\tau_{a}} \ltransition{} {\tau_{b}} \ltransition{} {\tau_{c}} \ltransition{} {\tau_{a}} S'' $ for some $S''$. In this sequence \emph{client} acquires parent's name by $ploc(x)$ and sends her source and the target locations to \emph{server} via $a$. The \emph{server} in response calculates the \emph{path (m,n)} between the two locations and delivers it back to the \emph{client}. In this particular case, the path calculated from $m$ to $n$ is $out \; m. in \; n$. Now the system has the form $S'' \equiv  \nu abc \;( sm[m[ \emph{client}[out \; m. in \; n.C' \mid  \emph{pda}[P']]] \mid n[\;]] \mid  server[S'] )$. After executing \emph{out m.in n} the final state of the system becomes $ 
\nu abc \;(  sm[m[ \; ] \mid n[\emph{client}[C' \mid \emph{pda}[P']]]] \mid  server[S'] )$, and is represented on the right hand side of Figure \ref{fig:mall}.

\section{Conclusion and Related Work} \label{sec:conclusion}

We have proposed CMC for the modelling of mobility, communication and context-awareness in the setting of ubiquitous computing. The notion of ambients mobility has been modelled by the \emph{in $n_B$} and \emph{out $n_B$} capabilities \cite{ambC}. A new form of global communication has been introduced in CMC which is similar to that in Milner's CCS. Ambient's name has been tagged with the set of ports which are functioning as a restriction on global communication, specified at the level of ambients.  A labelled transition system semantics has been developed, where $P \transition {}{\tau} Q$ represents not only a binary communication of processes as in CCS but also the ambients' mobility steps by means of their $in \; n_B$ and $out \; n_B$ capabilities. This has been achieved by additional labels and specialised transitions from processes to the so-called outcomes which are either processes or concretions.

Recently, a number of variants of MA have been introduced. \emph{Boxed Ambients (BA)} \cite{BA} inherits mobility primitives, namely the $in$ and $out$ capabilities from Mobile Ambients and introduce a direct communication method between parent and child. \emph{Channel Ambient calculus (CA)} \cite{CA} is a variant of Boxed Ambients. In CA, channels are defined as a first class objects and the communication is either between parent and child or between siblings. To the best of our knowledge, the ambient calculi do not support a direct interaction of an agent with a subagent inside another agent. Communication can only happen between the two adjacent agents, namely communication between parent and child or between siblings. CMC has introduced a new form of global communication by defining ambients as $m_A[P]$, where $m$ is the name of the ambient, $A$ is the set of ports that $m$ is allowed to communicate on, and $P$ is an executing agent.

Poslad in \cite{7P} addressed a number of theoretical concepts in the context of ubiquitous computing. In ubiquitous computing setting computations could be mobile and context-aware as, for example, in \cite{AI1,Context}. Satoh has researched spatial organisation of systems \cite{15IS,6IS} and concluded that technological advancements have enabled computing devices to become aware of their surroundings. Location-awareness has turned out to be useful in many applications, in particular, in determining position, navigation, tracking, and monitoring of ubiquitous computing devices. The notion of bigraph has been introduced by Milner in \cite{BG} with the idea of presenting two independent structures on the same set of nodes. A bigraph is a mathematical structure consisting of a place graph and a link graph with common nodes.
%The theory of bigraphical reactive systems \cite{11M} is based on a graphical model of mobile computation that emphasizes both locality and connectivity.
Process calculi and behavioural equivalences have led to an approach in bigraph theory somewhat different from the well-known tradition of graph rewriting \cite{Bigraph3}. Leonhardt \cite{LocModels} classified location models into geometric and symbolic models. In geometric models locations are represented as coordinates systems, whereas symbolic location models use the notion of place and labelling the locations. We use the notion of place to model location, and represent the structure of our system by a hierarchical space tree. The nodes represent the places, objects or computing devices, whereas the edges represent the containment relations between objects. Each node or object is represented by named ambient, which may contain nested ambients inside, as in \cite{ambC}.

\emph{A Calculus of Context Aware Ambients (CCA)} \cite{CCA2} describes the context-awareness requirements of the mobile systems. It introduces the notion of context expression that constraints the capabilities. We also add a context-awareness mechanism to our calculus by introducing two capabilities to it. The new capability \emph{ploc(x).P} allows an ambient to acquire the name of its parent and pass it as $x$ to $P$, whereas \emph{sloc(x).P} enquires the sibling's name of an ambient. Conversation Calculus \cite{CC, CC2} is designed for expressing and analysing service based systems. It proposes a spatial communication topology where conversation contexts are used as message exchange patterns. The construct \emph{here(x)} that allows access to the conversation medium in Conversation Calculus is similar to the \emph{ploc(x)} and \emph{sloc(x)} capabilities of our calculus. %The capabilities \emph{ploc(x)} and \emph{sloc(x)} enable ambients to be aware of their current locations and surroundings respectively.
These capabilities are not precisely used for only communication, whereas in Conversation Calculus conversation contexts are proposed as communication medium that controls information sharing among processes. Sessions \cite{sessions98} introduce a communication context among various partners to exchange messages based on previously agreed scheme, and sessions of specific patterns are introduced to express communication primitives. In CMC, we have modelled physical contexts and have intuitively used ambients to represent the structures. The systematic addition of context-awareness primitives smoothly increases the expressiveness power of the calculus.

In past few years, several operational semantics have been developed for MA and its variants as, for example, in \cite{BSTSA,BTMA2,ConcretionMA}. The authors in \cite{BSTSA} introduce a labelled transition system based operational semantics, and a labelled bisimulation equivalence which is proved to coincide with reduction barbed congruence. We also develop a labelled transition semantics and prove that the semantics coincides with the standard reduction semantics. Our labelled transition semantics is inspired by that in \cite{BSTSA}. The main difference is that we do not use the co-capabilities, hence preserving the standard MA semantics.
We have defined barbed bisimulation and congruence, and capability barbed bisimulation and congruence and have showed that the respective congruence relations of the two forms of barbs coincide. The notion of behavioural equivalence and the proof method for establishing the equivalence is inspired by that in \cite{BSTSA}. The authors in \cite{BSTSA}, use co-actions and passwords that help them in proving their results, whereas the use of global communication in CMC is fundamental in proving the results. The expressiveness and usefulness of the calculus has been illustrated by presenting intelligent hospital and interactive shopping mall case studies, where the relevant constructs are used to model various features of the calculus.

%\nocite{*}
\bibliographystyle{eptcs}
\bibliography{generic}
\newpage
\begin{appendices}
%\chapter{} \thispagestyle{empty} \label{appendix:A}

\section{SOS Rules for communication} \label{Appendix:A}
\begin{table}[htbp]
$
\begin{array}{llll}
\\
\AxiomC{}
	\LeftLabel{(Input)}
	\RightLabel{$(v \in V)$}
	\UnaryInfC{$a(z).P \transition{}{a(v)}{} P\lbrace v/z \rbrace$}
	\DisplayProof &
\quad \AxiomC{}
	\LeftLabel{(Output)}
	\UnaryInfC{$\overline{a}(x). P \transition{}{\overline{a}(x)}{} P$}
	\DisplayProof
\\[4pt]
	\AxiomC{$P \transition{}{\alpha}{} P'$}
	\LeftLabel{(Res-Act)}
	\RightLabel{ ($ a \notin \emph{fn}(\alpha) $)}
	\UnaryInfC{$(\nu a)P \transition{}{\alpha}{} (\nu a)P'$}
	\DisplayProof &
\quad \mbox{(Sum)} \displaystyle
 \frac{\transition{P}{\alpha}{P'}}{\transition{P + Q}{\alpha}{P'}}
  \\[4pt]
  \AxiomC{$P \transition{}{a(x)}{}P'$} \AxiomC{$Q\transition{}{\overline{a}(x)}{} Q'$}
	\LeftLabel{(Par-Com)}
	\BinaryInfC{$P\mid Q \transition{}{\tau}{} P'\mid Q'$}
	\DisplayProof &
  \quad \mbox{(Par-Act)} \displaystyle
 \frac{\transition{P}{\alpha}{P'}}{\transition{P\mid Q}{\alpha}{P'\mid Q}} 
 \\[4pt]
 \mbox{(Rel)} \displaystyle\frac{\transition{P}{\alpha}{P'}}{\transition{P[f]}{f(\alpha)}{P'[f]}}
 &
  \quad \mbox{(Const)} \displaystyle \displaystyle \frac{\transition{P}{\alpha}{P'}}{\transition{A}{\alpha}{P'}} \quad (A\dff P)
 \\[4pt]

   \AxiomC{$P\equiv Q \quad \transition{Q}{l}{Q'}\quad Q' \equiv P' $}
	\LeftLabel{(Struct)}
	\UnaryInfC{$\transition{P}{l}{P'}$}
	\DisplayProof
\\\\	
%	\multicolumn{2}{l}{
%	\AxiomC{$\transition{P}{\alpha}{P'}$}
%	\LeftLabel{(Global-Com)}
%	\RightLabel{ $( \cond (\alpha=a(x) \; or \; \alpha = \overline{a}(x)) \; then \;a \in A)$}
%	\UnaryInfC{$\transition{m_{A}[P]}{\alpha}{m_{A}[P']}$}
%	\DisplayProof 	
%	} & 
%	\\\\
\end{array}
$
\label{tab:sosCom}
\end{table}

\end{appendices}
\end{document}